\documentclass{llncs}

\usepackage{amssymb}
\usepackage{textcomp}
\usepackage{comment}
\usepackage{color}
\usepackage{graphicx}
\usepackage{amsmath}

\newcommand{\Endproof}{\hfill$\Box$}

\newcommand{\ket}[1]{|#1\rangle}

\DeclareMathOperator{\exv}{\mathbb{E}\textrm{ }}

\begin{document}

\title{Quantum versus Classical Online Streaming Algorithms with Logarithmic Size of Memory}
\author{Kamil~Khadiev$^{1,2,3}$ \and Aliya~Khadieva$^{1,2}$\and Dmitry Kravchenko$^{1}$ \and Alexander Rivosh$^{1}$ \and Ramis Yamilov$^{4}$ \and Ilnaz Mannapov$^{2,3}$}

\institute{
University of Latvia, Riga, Latvia
\and
Kazan Federal University, Kazan, Russia
  \and
Smart Quantum Technologies Ltd., Kazan, Russia
\and Yandex Ltd., Saint Petersburg, Russia 
                \\ \email{kamilhadi@gmail.com, aliya.khadi@gmail.com, kdmitry@gmail.com, alexander.rivosh@lu.lv, ramis.yamilov@gmail.com, ilnaztatar5@gmail.com}
}

\maketitle

\begin{abstract}We consider online algorithms with respect to the competitive ratio. Here, we investigate quantum and classical one-way automata with non-constant size of memory (streaming algorithms) as a model for online algorithms. We construct problems that can be solved by quantum online streaming algorithms better than by classical ones in a case of logarithmic or sublogarithmic size of memory. 
 \\
\textbf{Keywords:} quantum computation, online algorithm, automaton, streaming algorithm, online minimization problem, branching program, BDD
\end{abstract}

%---------------------------------------------------------------
% Introduction
%---------------------------------------------------------------
\section{Introduction}

Online algorithms are a well-known computational model for solving optimization problems.
%The defining property of this model is as follows. 
The peculiarity is that the algorithm reads an input piece by piece
and should return an answer piece by piece immediately, even if an answer can
depend on future pieces of the input. The algorithm should return an answer for minimizing an objective function (the cost of an output). The most standard method to define the effectiveness is the competitive ratio \cite{st85,kmrs86}. 
%It is the ratio of the cost of the online algorithm's solution and the cost of a solution of an optimal offline algorithm.
%
Typically, online algorithms have unlimited computational power.
%The main restriction is a lack of knowledge on future input variables.
At the same time, it is interesting to solve online minimization problems in a case of a big input stream such that the stream cannot be stored completely in the memory. In that case, we can consider automata (streaming algorithm) as online algorithms. This model was explored in \cite{bk2009,gk2015,blm2015,kkm2018}. We are interested in {\em quantum online algorithms}. This model was introduced in \cite{kkm2018} and discussed in \cite{aakv2018}.
It is known that quantum online algorithms can be better than classical ones in the case of sublogarithmic size of memory \cite{kkm2018}. Here, we consider logarithmic size of memory (polynomial number of states) that is more common memory restriction for streaming models. In this case, only quantum online algorithms with repeated test were considered in \cite{y2009}. In this paper, we focus on online streaming algorithms (one-way automata for online minimization problems) that read an input only once. %
The question of comparing quantum and classical models was explored for streaming computation models (OBDDs and automata), but not for the {\em online} streaming algorithms \cite{l2009,gkkrw2007,agkmp2005,agky16,ss2005,kk2017,aakk2018,g15,AY12,AY15,l2006,gy2017,gy2018,gy2015}.

We present the ``Black Hats Method'' for constructing hard online minimization problems. We use it and construct  problems that separate power of quantum algorithms  from classical ones. Suppose that algorithms use only $O(\log n)$ bits of memory ($n^{O(1)}$ states), where $n$ is the length of the input.
%
%\begin{itemize}
%\item 
%$\bullet$
There is a problem that has a quantum online streaming algorithm with a better competitive ratio than any classical (randomized or deterministic) online streaming algorithms. The problem is based on the $R$ function from \cite{ss2005}.
%
%$\bullet$ 
There is a problem that has quantum and randomized online algorithms with a better competitive ratio than any deterministic online algorithm. The problem is based on the Equality function and results from \cite{akv2008}.
%\end{itemize}
%
For both cases, the quantum online streaming algorithms (with $O(\log n)$ qubits) have a better competitive ratio than any deterministic online algorithm with unlimited computational power.

%Additionally, we show new results in case of sublogarithmic memory.
%
% \begin{itemize}
%\item
%$\bullet$ 
Suppose that the algorithms use a constant size of memory (constant number of states). There is a problem that has a quantum online streaming algorithm with singel qubit with a better competitive ratio than any classical online streaming algorithm. The problem is based on the $PartialMOD$ function from \cite{AY12,agky14,agky16} and the ``Black Hats Method''.

  The paper is organized in the following way. Definitions are in Section \ref{sec:prlmrs}. The Black Hats Method is described in Section \ref{sec:bhm}. Quantum and randomized vs. deterministic online streaming algorithms are discussed in the first part of Section \ref{sec:application}; the second part contains results on quantum vs. classical models.
%%%%%%%%%%%%%%%%%%%%%%%%%%%%%%%%%%%%%%%%%%%%
%              Preliminaries               %
%%%%%%%%%%%%%%%%%%%%%%%%%%%%%%%%%%%%%%%%%%%%
\section{Preliminaries}\label{sec:prlmrs}
%We give  all following definitions with respect to \cite{k2016,kkm2018,y2009}.
%
{\bf An online minimization problem} consists of a set $\cal{I}$ of inputs and a cost function. Each input $I = (x_1, \dots , x_n)$ is a sequence of requests, where $n$ is a length of the input $|I|=n$. Furthermore, a set of feasible outputs (or solutions) ${\cal O}(I)$ is associated with each $I$; an output is a sequence of answers $O = (y_1, \dots, y_n)$. The cost function assigns a positive real value $cost(I, O)$ to $I\in{ \cal I}$ and $O\in{\cal O}(I)$. An optimal  solution for $I\in{\cal I}$ is $O_{opt}(I)=argmin_{O\in{\cal O}(I)}cost(I,O)$.

Let us define an online algorithm for this problem.
%as an algorithm which gets requests $x_i$ from $I=(x_1,\dots,x_n)$ one by one and should return answers $y_i$ from $O=(y_1,\dots,y_n)$ immediately, even if an optimal solution can depend on future requests.
%
{\bf A deterministic online algorithm}  $A$ computes the output sequence $A(I) = (y_1,\dots , y_n)$ such that $y_i$ is computed from $x_1, \dots , x_i$.  
%
 % This setting can also be regarded as a request-answer game: an adversary generates requests, and an online algorithm has to serve them one at a time \cite{a1996}.  
% 
%   We use the competitive ratio as the main measure of quality for the online algorithm. It is the ratio of the cost of the algorithm's solution and the cost of a solution of an optimal offline algorithm in the worst case. 
%
We say that $A$ is $c$-{\em competitive} if there exists a constant $\alpha\geq 0$ such that, for every $n$ and for any input $I$ of size $n$, we have: $cost(I,A(I)) \leq c \cdot cost(I,Opt(I)) + \alpha,$ where $Opt$ is an optimal offline algorithm for the problem and $c$ is the minimal number that satisfies the inequality. Also we call $c$ the {\bf competitive ratio} of $A$. If $\alpha = 0, c=1$, then $A$ is optimal.

{\bf A randomized online algorithm} $R$ computes an output sequence
$R^{\psi}(I) = (y_1,\cdots, y_n)$ such that $ y_i$ is computed from $\psi, x_1, \cdots, x_i$, where $\psi$ is a content of the random tape, i. e., an infinite binary sequence, where every bit is chosen uniformly at random and independently of all the others. By $cost(I,R^{\psi}(I))$ we denote the random variable expressing the cost of the solution computed by $R$ on $I$.
$R$ is $c$-competitive in expectation if there exists a constant $\alpha\geq 0$ such that, for every $I$, $\exv[cost(I,R^{\psi}(I))] \leq c \cdot cost(I,Opt(I)) + \alpha$.

We use one-way automata for online minimization problems as online algorithms with restricted memory. In the paper, we use the terminology for branching programs \cite{Weg00} and algorithms. We say that an automaton computes Boolean function $f_m$ if for any input $X$ of length $m$, the automaton returns result $1$ iff $f(X)=1$. Additionally, we use the terminology on memory from algorithms. We say that an automaton has $s$ bits of memory if it has $2^s$ states. Let us present the definitions of automata that we use. A {\bf deterministic automaton} with $s=s(n)$ bits
of memory that process input $I=(x_1,\dots,x_n)$ is a $4$-tuple $(d_0,D,\Delta,Result)$. The set $D$ is a set of states, $|D|=2^s$, $d_0\in D$ is an initial state. $\Delta$ is a transition function $\Delta:\{0,\dots,\gamma-1\}\times D\to D$, where $\gamma$ is a size of the input alphabet. $Result$ is an output function $Result:D\to\{0,\dots,\beta-1\}$, where $\beta$ is a size of the output alphabet. The computation starts from the state $d_0$. Then on reading an input symbol $x_j$ it changes the current state $d\in D$ to $\Delta(x_j,d)$. In the end of computation, the automaton outputs $Result(d)$. A {\bf probabilistic automaton} is a probabilistic counterpart of the model. It
chooses from more than one transitions in each step such that each transition is associated with a probability. 
Thus, the automaton can be in a probability distribution over states during the computation. A total probability must be $1$, i.e., a probability of outgoing transitions from a single state must be $1$. Thus, a probabilistic automaton returns some result for each input with some probability. For $v\in\{0,\dots,\beta-1\}$, the automaton returns a result $v$ for an input, with bounded-error if the automata returns the result $v$ with probability at least $1/2 + \varepsilon$ for some $ \varepsilon \in (0,1/2] $. The automaton computes a function $f$ with bounded error  if it returns $f(X)$ with bounded error for each $X\in\{0,\dots,\gamma-1\}^n$. The automaton computes a function $f$ exactly if $\varepsilon=0$. 

A {\bf deterministic online streaming algorithm} with $s=s(n)$ bits
of memory that process input $I=(x_1,\dots,x_n)$ is a $4$-tuple $(d_0,D,\Delta,Result)$. The set $D$ is s set of states, $|D|=2^s$, $d_0\in D$ is an initial state. $\Delta$ is a transition function $\Delta:\{0,\dots,\gamma-1\}\times D\to D$. $Result$ is an output function $Result:D\to\{0,\dots,\beta-1\}$. The computation starts from the state $d_0$. Then on reading an input symbol $x_j$ it change the current state $d\in D$ to $\Delta(x_j,d)$ and outputs $Result(d)$.
A {\bf randomized online streaming algorithm}  has a similar definition, but with respect to definitions of the corresponding model of online algorithms.

{\bf Comment.} Note that any online algorithm can be simulated by online streaming algorithm using $n$ bits of memory.

Let us consider a {\bf quantum online streaming algorithm}. The good sources on quantum computation are \cite{nc2010,AY15}.
For some integers $n>0 $, a quantum online algorithm $ Q$ with $q$ qubits is defined
on input $I=(x_1,\dots,x_n)\in\{0,\dots,{\gamma-1}\}^n $ and outputs 
$(y_1,\dots,y_{n})\in\{0,\dots,\beta-1\}^{n}$.
A memory of the quantum algorithm is a state of a quantum register of $q$ qubits. In other words, the computation of $Q$ on an input $I$ can be traced by a $ 2^q$-dimensional vector from Hilbert space over the field of complex numbers. The initial state is a given $ 2^q$-vector $ |\psi\rangle_0$. In each step $j\in\{1,\dots,n\}$ the input variable $ x_{j} $ is tested and then a unitary $ 2^q\times2^q$-matrix $ G^{x_{j}}$ is applied:  
$
    |\psi\rangle_j = G^{x_{j}} (|\psi\rangle_{j-1}),
$ 
where  $ |\psi\rangle_j $ represents the state of the system after the $ j$-th step.
Depending on an input symbol, the algorithm can measure one or more quantum bits. If the outcome of the measurement is $u$, then the algorithm continues computing from a state $|\psi(u)\rangle$ and tha algorithm can output $Result(u)$ on this step. Here  $Result:
\{0,\ldots,2^{q}-1\}\to\{0,\dots, \beta-1\}$ is a function that converts the result of the measurement to an output variable. The algorithm $Q$ is $c$-competitive in expectation if there exists a non-negative constant $\alpha$ such that, for every $I$, $\mathbb{E}[cost(I,Q(I))] \leq c \cdot cost(I,Opt(I)) + \alpha$.

Let us describe a measurement process. Suppose that $Q$ is in a state $|\psi\rangle=(v_1, \dots, v_{2^q})$ before a measurement and the algorithm measures the $i$-th qubit. Suppose states with numbers $a^0_1,\dots,a^0_{2^{q-1}}$ correspond to $0$ value of the $i$-th qubit, and states with numbers $a^1_1,\dots,a^1_{2^{q-1}}$ correspond to $1$ value of the qubit. Then the result of the qubit's measurement is $1$ with probability $pr_1= \sum_{j=1}^{2^{q-1}}|v_{a^1_j}|^2$ and $0$ with probability $pr_0=1-pr_1$. If
the algorithm measures $v$ qubits on the $j$-th step, then $u \in\{0,\dots,2^v-1\}$ is an outcome of the measurement.

A quantum automata have the similar definition, but it returns $Result(u)$ in the end of the computation. A definition of a function computing is similar to the probabilistic case. See \cite{AY15} for more details on quantum automata.

In the paper we use results on id-OBDD. This model can be considered as an automaton with transition function that depends on position of input head. You can read more about classical and quantum id-OBDDs in
\cite{Weg00,ss2005,agkmp2005,agky14,agky16,kk2017}. Formal definitions of
this model is in Appendix \ref{apx:obdd}. The following relations 
between models are folklore:

\begin{lemma}\label{lm:rel-obdd-sa}
 If a quantum (probabilistic) id-OBDD  $P$ of width $2^w$ computes a Boolean function $f$, then there is a quantum (probabilistic) automaton computing $f$ that uses $w + \lceil\log_2 n \rceil$ qubits (bits) of memory.
If any deterministic (probabilistic) id-OBDD  $P$ computing a Boolean function $f$ has a width at least $2^w$, then any deterministic (probabilistic) automaton computing $f$ uses at least $w$ bits of memory.
\end{lemma}

\section{The Black Hats Method for Constructing Online Minimization Problems}\label{sec:bhm}
Let us define a method which allows us to construct hard online minimization problems. We call it the ``Black Hats Method''.
In this paper, we say a Boolean function $f$, but in fact we consider a family of Boolean functions $f=\{f_1,f_2,\dots\}$, for $f_m:\{0,1\}^m\to\{0,1\}$. We use notation $f(X)$ for $f_m(X)$ if the length of $X$ is $m$ and if it is clear from the context.

Suppose we have a Boolean function $f$ and positive integers $k,r,w,t$, where $k$ mod $t=0$, $r<w$. We define the online minimization problem $BH^t_{k,r,w}(f)$ as follows.
We have $k$ guardians and $k$ prisoners. They stay one by one in a line like $
G_1 P_1 G_2 P_2 \dots$, where $G_i$ is a guardian, $P_i$ is a prisoner. The prisoner
$P_i$ has an input $X_i$ of length $m_i$ and computes a function $f_{m_i}(X_i)$. If the
result is $1$, then the prisoner paints his hat black; otherwise, he paints it
white. Each guardian wants to know whether a number of following black hats is odd. We can say that the $i$-th guardian wants to compute $\bigoplus_{i=j}^k f_{m_i}(X_i)$. 
Formally, the problem is
\begin{definition}[The Black Hats Method]\label{def:bhm}
We have a Boolean function $f$. Then an online minimization problem $BH^t_{k,r,w}(f)$, for positive integers $k,r,w,t$, where $k$ mod $t=0$, $r<w$ is the following.
 Suppose we have an input $I=(x_1,\dots,x_n)$
and $k$ positive integers $m_1,\dots,m_k$, where $n=\sum_{i=1}^{k}(m_i+1)$. 
Let $I$ be such that $I=(2,X_1,2,X_2,2,X_3,2,\dots,2,X_k)$, where
$X_i\in\{0,1\}^{m_i}$, for $i\in\{1,\dots,k\}$. Let $O\in{\cal O}(I)$ and let
$O'=(y_1,\dots,y_k)$ be answer variables corresponding to input variables with
value $2$ (output variables for guardians). An variable $y_j$ corresponds to $x_{i_j}$, where $i_j=j+\sum_{r=1}^{j-1}m_r$. Let 
$g_j(I)=\bigoplus_{i=j}^k f_{m_i}(X_i)$.
We separate all answer variables $y_i$ into $t$ blocks of length $ z=k/t $. A
cost of the $i$-th block is $c_i$. Here $c_i=r$, if $y_j=g_j(I)$ for $j\in\{(i-1)z+1,\dots,i\cdot z\}$; and $c_i=w$, otherwise.
The cost of the whole output is $cost^t(I,O)=c_1+\dots+c_t$.
\end{definition} 

Let us discuss the method. If we have a quantum or randomized streaming algorithm for $f$ using a small amount of memory, then it is enough to guess the result of the first guardian to solve the problem.
Moreover, if there is no randomized streaming algorithm with small memory for $f$, then we cannot solve   $BH^t_{k,r,w}(f)$. The only way to reduce the competitive ratio is guessing the answers. Suppose we have a quantum streaming algorithm that uses small memory for $f$. Then long blocks can increase the gap between the competitive ratios of the randomized and quantum algorithms because all guardians inside a block should return right answers.
We have a similar situation with deterministic online algorithms. In that case, we cannot guess answers; and we have a more significant gap between competitive ratios for quantum and deterministic algorithms. 
The construction of the problem $BH^t_{k,r,w}(f)$ allows us to get a good competitive ratio by guessing only one bit; for example, this effect cannot be achieved by considering independent instances of a Boolean function $f$. 
These results are presented formally in theorems of this section. 
\begin{theorem}\label{th:bh-lower} Let $s$ be a positive integer, let $f$ be a Boolean function. Suppose there is no deterministic  automaton for $f$ that uses at most $s$ bits of memory. Then there is no $c$-competitive deterministic online streaming algorithm for $BH^t_{k,r,w}(f)$ that uses $s$ bits of memory, where $c<\frac{w}{r}$.
\end{theorem}
\begin{proof}
Let us consider any online streaming algorithm $A$ for the $BH^t_{k,r,w}(f)$  problem that uses at most $s$ bits of memory. Suppose that $A$ returns $y_1$ as an answer of the first guardian.
Let us prove that there are two parts of the input $X_1^0,X_1^1\in\{0,1\}^{m_1}$  such that $A$ returns the same value $y_2$ for both, but $f(X_1^0)=0, f(X_1^1)=1$.

Assume that there is no such triple $
(y_2,X_1^0, X_1^1)$. Then it means that we can construct an automaton $A'$ that uses $s$ bits of memory and has the following property:  $A'(X_1')=A'(X_1'')$ iff $f(X_1')=f(X_1'')$, for any $X_1',X_1''\in\{0,1\}^{m_1}$. The automaton $A'$ emulates the algorithm $A$. Therefore, $A'$ computes $f$ or $\lnot f$. In the case of $\lnot f$, we can construct $A''$ such that $A''(X_1)=\lnot A'(X_1)$. It is a contradiction with the claim of the theorem.
By the same way we can show that we have similar triples $
(y_{i+1},X_i^0, X_i^1)$ for $i\in\{2,\dots,k\}$.

Let us choose $\sigma_i=y_i\oplus 1 \oplus \bigoplus_{j=i+1}^{k}\sigma_{j}$, for $i\in \{1,\dots, k\}$.
Let us consider an input $I_A=2X_1^{\sigma_1}2X_2^{\sigma_2}2\dots2X_k^{\sigma_k}$.
An optimal offline solution is $(g_1,\dots,g_k)$ where
$g_i=\bigoplus_{j=i}^k \sigma_j$.

Let us prove that $g_i\neq y_i$ for each $i\in\{1,\dots,k\}$. We have
$\sigma_i=y_i\oplus 1 \oplus \bigoplus_{j=i+1}^{k}\sigma_{j}$, therefore
$y_i=\sigma_i\oplus 1 \oplus \bigoplus_{j=i+1}^{k}\sigma_{j}=1 \oplus \bigoplus_{j=i}^{k}\sigma_{j}=1\oplus g_i$, so $y_i = \lnot g_i$.

Hence, all answers are wrong and $cost^t(I_A,A(I_A))=tw$. So the competitive ratio $c$ cannot be less than $ tw/(tr)=w/r$.
\Endproof
\end{proof}
\begin{theorem}\label{th:bh-rand-lower}
Let $s$ be a positive integer, let $f$ be a Boolean function. Suppose there is no probabilistic automaton that uses at most $s$ bits of memory and computes $f$ with bounded error. Then there is no $c$-competitive in expectation randomized online streaming algorithm $A$ for $BH^t_{k,r,w}(f)$ that uses $s$ bits of memory, where $c<2^{-z}+(1-2^{-z})w/r$, $z=k/t$. 
\end{theorem}
\begin{proof}We can show that an algorithm $A$ cannot compute an answer of a guardian $y_i$ with bounded error. The idea of the proof is similar to the proof of Theorem \ref{th:bh-lower}. It means that the only way to answer is guessing $y_i$ with probability $0.5$. We achieve the claimed competitive ratio for this strategy.  
If the algorithm wants to get a cost $r$ for a block, then it should guess all output bits of the block. So, the cost of the $i$-th block is $c_i=(1-2^{-z})w+2^{-z}r$. Therefore, $cost^t(I_A,A(I_A))=t \cdot((1-2^{-z})w+2^{-z}r)$. Hence, the algorithm is $c$ competitive in expectation, for $c\geq t \cdot((1-2^{-z})w+2^{-z}r)/(tr)=2^{-z}+(1-2^{-z})w/r$.
\Endproof
\end{proof}
The following theorem is a bound on the competitive ratio in the case of an unlimited computational power for a deterministic online algorithm. 
\begin{theorem}\label{th:bh-general-d-lower} 
There is no $c$-competitive deterministic online algorithm  for $BH^t_{k,r,w}(f)$, where $c< \big(\lfloor (t+1)/2 \rfloor \cdot w + (t - \lfloor (t+1)/2 \rfloor\big)\cdot r)/(tr)$.
\end{theorem}
\begin{proof} We can construct an input $I$ such that at least $\lfloor (k+1)/2 \rfloor $ guardians return wrong answers. 
Suppose that the algorithm $A$ receives the input $I=(2,X_1,2,X_2,2,\dots,2,X_k)$, where $X_i,\in\{0,1\}^{m_i}$. 

Let us choose $X_i$ such that $f(X_i)=0$ for $i \in \{1,\dots,k-1\}$.
Then $A$ receives the part $(2,X_1,\dots,2,X_{k-1},2)$ of the input and returns $(y_1,\dots,y_{k})$. 

Let $b=1$, if $y_1+\dots+y_k\geq \lfloor (k+1)/2 \rfloor$; and $b=0$, otherwise. Then we choose $X_k$ such that $f(X_k)\neq b$. In that case $g_1=\dots=g_k= f(X_k)\neq b$, where $g_i=\bigoplus_{j=i}^k f(X_j)$. Therefore, at least $\lfloor (k+1)/2 \rfloor$ guardians return wrong answers.

The worst case is the first  $\lfloor (k+1)/2 \rfloor$ guardians return wrong answer. So, at least $\lfloor (t+1)/2 \rfloor$ blocks will be ''wrong'', and the algorithm is $c$-competitive, for $c\geq \big(\lfloor (t+1)/2 \rfloor \cdot w + (t - \lfloor (t+1)/2 \rfloor\big)\cdot r)/(tr)$.
\Endproof
\end{proof}
\begin{corollary}\label{cr:nobenefit-advice}
There is no deterministic online algorithm $A$ for $BH^1_{k,r,w}(f)$ that is $c$-competitive, for $c< w/r$.
\end{corollary}

\begin{theorem}\label{th:pqalgorithm} 
Let $s$ be a positive integer, let $f$ be a Boolean function. Suppose we have a quantum (probabilistic) automaton $R$ that computes $f$ with bounded error $\varepsilon$ using $s$ qubits (bits) of memory, where $0\leq \varepsilon<0.5$. Then there is a quantum (randomized) online streaming algorithm $A$ for $BH^t_{k,r,w}(f)$ that uses at most $s+1$ qubits (bits) of memory and has expected competitive ratio %
$c\leq \left(0.5(1-\varepsilon)^{z-1}\cdot(r-w) + w\right)/r$, $z=k/t$.
\end{theorem}
\begin{proof}Let us present the randomized online streaming algorithm $A$:

{\bf Step $1$.} The algorithm \(A\) guesses \(y_1\) with
equal probabilities and stores it in a bit \(p\). So, $p=1$ or $p=0$ with probability $0.5$. Then $A$ returns $y_1=p$.

{\bf Step $2$.} The algorithm reads $X^1$ and computes $p=p\oplus R(X^1)$, where $R(X^1)$ is a result of computing $R$ on $X^1$. Then $A$ returns $y_2=p$.

{\bf Step $i$.} The algorithm reads $X^{i-1}$, computes $p=p\oplus R(X^{i-1})$ and returns $y_i=p$.

{\bf Step $k$.} The algorithm reads and skips $X^{k}$. The algorithm $A$ does not need these variables, because it guesses $y_1$ and using this value we  can obtain $y_{2},...,y_{k}$ without $X^{k}$.

Let us compute a cost of the output for this algorithm.
Let us consider a new cost function $cost'(I,O)$. For this function, a ``right''  block costs $1$ and a ``wrong'' block costs $0$. So, $cost^t(I,O)=(r-w)\cdot cost'(I,O) + tw$. Let us compute $\exv[cost'(I,O)]$. We recall  that the problem has $k$ guardians, $t$ blocks and $z=k/t$. 

Firstly, let us compute $p_i$ the probability that block $i$ is a ``right'' block (costs $1$). 
Let $i=1$. So, if the $i$-th block is ``right'', then all $z-1$  prisoners inside the block return right answers and a guess of the first guardian is right. A probability of this event is $p_1=0.5\cdot (1-\varepsilon)^{z-1}$.

Let $i>1$. If the $i$-th block is ``right'', then two conditions should be true:

(i) All $z-1$ prisoners inside the block should return right answers. 

(ii) If we consider a number of preceding guardians that return wrong answers plus $1$ if the preceding prisoner has an error. Then this number should be even.

A probability of the first condition is $(1-\varepsilon)^{z-1}$. Let us compute a probability of the second condition.

Let $E(j)$ be the number of errors before the $j$-th guardian. It is a number of errors for the previous prisoners plus $1$ if the guess of the first guardian is wrong. Let $F(j)$ be a probability that $E(j)$ is even. Therefore $1-F(j)$ is a probability that $E(j)$ is odd. If there is an error in a computation of the $(j-1)$-th prisoner, then $E(j-1)$ should be odd. If there is no error for the $(j-1)$-th prisoner, then $E(j-1)$ should be even. Therefore, $F(j)=\varepsilon (1-F(j-1)) + (1-\varepsilon)F(j-1)= F(j-1)(1-2\varepsilon) + \varepsilon$.  Note that the guess of the first guardian is right with probability $0.5$. Therefore, $F(1)=0.5$.

So,
$F(j)=F(j-1)(1-2\varepsilon) + \varepsilon = F(j-2)(1-2\varepsilon)^2 + (1-2\varepsilon)\varepsilon + \varepsilon=\dots =$

$=F(j-j+1)(1-2\varepsilon)^{j-1} + (1-2\varepsilon)^{j-2}\varepsilon + \dots + (1-2\varepsilon)\varepsilon + \varepsilon=F(1)\cdot(1-2\varepsilon)^{j-1} + \varepsilon\sum_{l = 0}^{j-2}(1-2\varepsilon)^{l}=$
$\frac{(1-2\varepsilon)^{j-1}}{2} +\frac{1-(1-2\varepsilon)^{j-1}}{2}=0.5$

Hence, $p_i =0.5\cdot (1-\varepsilon)^{z-1}$. 

Finally, let us compute the expected cost: 

$\exv[cost'(I,A(I))]=\sum_{i=1}^{t}\big(p_i\cdot 1 + (1-p_i)\cdot 0\big)=\sum_{i=1}^tp_i=0.5 \cdot (1-\varepsilon)^{z-1}\cdot t$.

Therefore, $\exv[cost^t(I,A(I))]= 0.5 \cdot (1-\varepsilon)^{z-1}\cdot t(r-w) + tw$.

Let us compute expected competitive ratio $c$:

$c\leq \left(0.5 \cdot (1-\varepsilon)^{z-1}\cdot t(r-w) + tw\right)/(tr)= \left(0.5 \cdot (1-\varepsilon)^{z-1}\cdot (r-w) + w\right)/r$

Let us present the quantum online streaming algorithm $A$:

{\bf Step 1.} The algorithm \(A\) guesses \(y_1\) with
equal probabilities and stores it in a qubit \(|p\rangle\): the algorithm initialize the qubit \(|p\rangle=\frac{1}{\sqrt{2}}|0\rangle+\frac{1}{\sqrt{2}}|1\rangle\). Then \(A\) measures $|p\rangle$ and returns a result of the measurement as \( y_1\).

{\bf Step 2.} The algorithm reads $X^1$ and computes $|p\rangle$ as a result of  CNOT or XOR of $|p\rangle$ and  $R(X^1)$, where $R(X^1)$ is the result of computation for $R$ on the input $X^1$. The algorithm $A$ uses a register $|\psi\rangle$ of $s$ qubits for processing $X^1$. Then the algorithm returns a result of a measurement for $|p\rangle$ as $y_2$. After that $A$ measures all qubits of $|\psi\rangle$ and sets $|\psi\rangle$ to $|0\dots 0\rangle$. The algorithm can do it because it knows a result of the measurement and can rotate each qubit such that the qubit becomes $|0\rangle$. 

{\bf Step $i$.} The algorithm reads $X^{i-1}$ and computes $|p\rangle$ as a result of  CNOT or XOR of $|p\rangle$ and  $R(X^{i-1})$. The algorithm $A$ uses the same register $|\psi\rangle$ of $s$ bits on processing $X^{i-1}$. Then $A$ returns a result of the measurement for $|p\rangle$ as $y_i$. After that the algorithm measures $|\psi\rangle$ and sets $|\psi\rangle$ to $|0\dots 0\rangle$.

{\bf Step $k$.} The algorithm reads and skips $X^{k}$. The algorithm $A$ does not need these variables, because it guesses $y_1$ and using this value we  can obtain $y_{2},...,y_{k}$ without $X^{k}$.

The bound on the expected competitive ratio for the quantum online streaming algorithm is the same as for randomized one.
\Endproof
\end{proof}
%\begin{proof} Here we present only a randomized algorithm $A$. See Appendix \ref{apx:pqalgorithm} for the quantum algorithm and bounds on the competitive ratio.
%
%\noindent
%{\bf Step $1$.} The algorithm \(A\) guesses \(y_1\) with
%equal probabilities and stores it in a bit \(p\). So, $p=1$ or $p=0$ with probability $0.5$. Then $A$ returns $y_1=p$.

%\noindent
%{\bf Step $2$.} The algorithm reads $X^1$ and computes $p=p\oplus R(X^1)$, where $R(X^1)$ is a result of computing $R$ on $X^1$. Then $A$ returns $y_2=p$.
%
%\noindent
%{\bf Step $i$.} The algorithm reads $X^{i-1}$, computes $p=p\oplus R(X^{i-1})$ and returns $y_i=p$.
%
%\noindent
%{\bf Step $k$.} The algorithm reads and skips $X^{k}$. The algorithm $A$ does not need these variables, because it guesses $y_1$; and using this value we  can obtain $y_{2},\dots,y_{k}$ without $X^{k}$.
%
%A quantum algorithm is similar, but has one additional action: it measures the whole quantum memory after each step and sets it to $|0\dots 0\rangle$ state before a next step. We can do such initialization because we know the result of measurement for the previous step. 
%\Endproof
%\end{proof}

%\subsection{Advice complexity}

\section{Application}\label{sec:application}
Let us discuss the applications of the Black Hats Method. In this section, we present examples of problems that allow us to show the benefits of quantum computing in the case of online streaming algorithms. 
Here we use results for OBDDs. See Appendix \ref{apx:obdd} for a definition of OBDD.
%\subsubsection{Quantum and Probabilistic vs. Deterministic Algorithms}
Recall that $BH^t_{k,r,w}(f)$ is a black hat problem for $k$ guardians, $t$ blocks of guardians, $r$ and $w$ are costs for a right and a wrong answers of a block, respectively, $z=k/t$ and $k$ mod $t=0$.

\subsection{Quantum and Probabilistic vs. Deterministic Algorithms.}
Let us apply the Black Hats Method from Section \ref{sec:bhm} to a Boolean function $EQ_n$ from \cite{akv2008}.
 The Boolean function $EQ_n:\{0,1\}^n\to\{0,1\}$ is such that $EQ(x_1,\dots, x_{n})=1$ if $(x_1,\dots x_{\lfloor n/2\rfloor})=(x_{\lfloor n/2\rfloor+1}, \dots x_{n})$; and $0$ otherwise.
It is known from \cite{akv2008,Fre79,agkmp2005} that there are quantum and probabilistic OBDDs that compute $EQ_n$ using linear width.  At the same time, any deterministic OBDD requires exponential width. Hence, we have the following property due to Lemma \ref{lm:rel-obdd-sa}.

\begin{lemma}\label{lm:bhe}
1. There are quantum and randomized automata that compute  $EQ_n$ using $O(\log n)$ qubits (bits) of memory with one-sided error $ \varepsilon$.
%\item 
2. There is no deterministic automaton that computes $EQ_n$ using $o(n)$ bits of memory.
%\end{itemize}
\end{lemma}
%\begin{proof}The claim follows from results for OBDDs \cite{akv2008,Fre79,agkmp2005} and Lemma \ref{lm:rel-obdd-sa}.
%\Endproof
%\end{proof}

Let us consider the $BHE^t_{k,r,w} =BH^t_{k,r,w}(EQ)$ problem. 
The following properties of the problem are based on Lemma \ref{lm:bhe} and Theorems \ref{th:pqalgorithm}, \ref{th:bh-lower}, \ref{th:bh-rand-lower}.

\begin{theorem}
Suppose $P^t=BHE^t_{k,r,w}$, $t\in\{1,\dots,k\}$, $k=O(\log n)^{O(1)}$, then
%\begin{enumerate}

%\item 
1. There is no $c$-competitive deterministic online streaming algorithm with $o(n)$ bits of memory that solves $P^t$, where $c<{\cal C}_1=\frac{w}{r}$.

%\item 
2. There is no deterministic online algorithm with unlimited computational power computing $P^1$ that is $c$-competitive, for $c<{\cal C}_1=w/r$. 

%\item 
3. There are quantum and randomized online streaming algorithms that use $O(\log n)$ qubits (bits) and solve $P^t$. These algorithms are $c$-competitive in expectation, where
$c\leq \left((1-\varepsilon)^{z-1}\cdot 0.5 \cdot (r-w) + w\right)/r<{\cal C}_1,{\cal C}_2$.
%\end{enumerate}
 %
\end{theorem}
\begin{proof}
Due to Lemma \ref{lm:bhr}, there is no deterministic automaton that computes $EQ_n$ in a case of  $o(n)$  bits of memory. 
So, if we use these properties and Theorems \ref{th:bh-lower}, \ref{th:bh-rand-lower}, \ref{th:bh-general-d-lower} and Corollary \ref{cr:nobenefit-advice}, then we obtain the claims 1 and 2 of the theorem.  
At the same time, we have quantum and randomized automata for $EQ_n$ with bounded error that uses $O(\log n)$ qubits. If we apply this property and Theorem \ref{th:pqalgorithm}, then we obtain the claim 3. 
\Endproof
\end{proof}

This theorem gives us the following significant results.
%
%\begin{enumerate}
%\item 
%Firstly, 
Quantum and randomized online streaming algorithms with logarithmic size of memory for  $BHE^1_{k,r,w}$ have better competitive ratios than any deterministic online algorithm without restriction on memory.
%\end{itemize} 
%\item
%
%Secondly, if we increase the number of advice bits for a deterministic online streaming algorithm for $BHE^t_{k,r,w}$, then the competitive ratio becomes smaller, in the case of polylogarithmic size of memory and $1<t\leq k/2$. At the same time, the competitive ratio is still larger than for the quantum and randomized online streaming algorithms. 
%\end{enumerate}
%\subsection{Quantum vs. Classical Algorithms}

%{\bf Quantum vs. Classical Algorithms.}
\subsection{Quantum vs. Classical Algorithms.}
\subsubsection{Polylogarithmic Memory}
%{\bf Polylogarithmic Size of Memory.}
We start by analyzing the model with polylogarithmic size of memory.
Let us apply the Black Hats Method from Section \ref{sec:bhm} to  a Boolean function $R_{\nu, l, m,u}:\{0,1\}^n\to\{0,1\}$ from \cite{ss2005}:
Let $|1\rangle, \dots , |u\rangle$ be the standard basis of $\mathbb{C}^u$. Let
$V_0$ and $V_1$ denote the subspaces spanned by the first and last $u/2$ of these basis vectors.
Let $0 < \nu < 1/\sqrt{2}$. The input for the function $R_{\nu, l, m,u}$  consists of $3l(m + 1)$ Boolean variables $a_{i,j}, b_{i,j}, c_{i,j}, 1 \leq i \leq l , 1 \leq j \leq m+1$, which are interpreted as universal $(\epsilon, l, m )$-
codes for three unitary $u \times u$-matrices A, B, C, where $\epsilon = 1/(3u)$. The function takes the
value $z \in \{0, 1\}$ if the Euclidean distance between $CBA|1\rangle$ and $V_z$ is at most $\nu$. Otherwise the function is undefined.
It is known from \cite{ss2005} that there is a quantum OBDD that computes $R_{\nu, l, m,u}$ using linear width.  At the same time, any deterministic or probabilistic OBDD requires exponential width. Therefore, we have the following result due to Lemma \ref{lm:rel-obdd-sa}.

\begin{lemma}\label{lm:bhr}
%\begin{itemize}
%\item 
1. There is a quantum automaton that computes $R_{\nu, l, m,u}$ with bounded error $\nu^2$. using $O(\log n)$ qubits. 
%\item 
2. there is no probabilistic automaton that computes $R_{\nu, l, m,u}$ with bounded error using $n^{o(1)}$ bits of memory .
%\end{itemize}
\end{lemma}
Let us consider the $BHR^t_{k,r,w,\nu, l, m,u} =BH^t_{k,r,w}(R_{\nu, l, m,u})$ problem.
\begin{theorem}
Suppose $P^t=BHR^t_{k,r,w,\nu, l, m,u}$,  $t\in\{1,\dots,k\}$, $k = (\log_2 n)^{O(1)}$; then
%\begin{enumerate}

%\item 
1. There is no $c$-competitive deterministic online streaming algorithm with $n^{o(1)}$ bits of memory that solves $P^t$, where $c<{\cal C}_1=\frac{w}{r}$.

%
%\item 
2. There is no deterministic online algorithm with unlimited computational power computing $P^1$ that is $c$-competitive, for $c<{\cal C}_1=w/r$. 

%\item 
3. There is no $c$-competitive in expectation randomized online streaming algorithm with $n^{o(1)}$ bits of memory that solves $P^t$, where $c<{\cal C}_2=2^{-z}+(1-2^{-z})\frac{w}{r}$.

%\item 
4. There is a quantum online streaming algorithm that uses $O(\log n)$ qubits and solves $P^t$. The algorithm $Q$ is $c$-competitive in expectation, where

$c\leq \left(\left(1-\nu^2\right)^{z-1}\cdot 0.5 \cdot (r-w) + w\right)/r<{\cal C}_1,{\cal C}_2,$.

%\item There is a quantum online streaming algorithm $A$ for $P^t$ such that the algorithm gets $1$ advice bit, uses $(\log_2 n)^{O(1)}$ qubit of memory and has the competitive ratio in expectation 

%$c\leq \left(0.5(1-\nu^2)^{z-1}\cdot \left(t +1 + \frac{v^t-v}{v-1} \right)(r-w) + tw\right)/(tr)<{\cal C}_1,{\cal C}_2,{\cal C}_3,{\cal C}_4$, for $v=(1-2\nu^2)^z$.
%\end{enumerate} 
\end{theorem}
\begin{proof}
Due to Lemma \ref{lm:bhr}, there is no probabilistic automaton that computes $R_{\nu, l, m,u}$ with bounded error in a case of $n^{o(1)}$ bits of memory. Therefore, we have a similar result in a deterministic case.
So, if we use these properties and Theorems \ref{th:bh-lower}, \ref{th:bh-rand-lower} and \ref{th:bh-general-d-lower} , then we obtain the claims 1,2 and 3 of the theorem.  
At the same time, we have quantum automaton for $R_{\nu, l, m,u}$ with bounded error $\nu^2$ using $O(\log n)$ qubits.  If we apply this property and Theorem \ref{th:pqalgorithm}, then we obtain the claim 4. 
\Endproof
\end{proof}

This theorem gives us the following important results.
%\begin{enumerate}
%\item 
%Firstly, 
There is a quantum online streaming algorithm with logarithmic size of memory for  $BHR^t_{k,r,w,\nu, l, m,u}$ having a better competitive ratio than
%\begin{itemize}
%\item 
(i) any classical (deterministic or randomized) online streaming algorithm with polylogarithmic size of memory;
%\item 
(ii) any deterministic online algorithm without restriction on memory.
%\end{itemize} 
%\item 
%Secondly, if we increase the number of advice bits for a classical online streaming algorithm for $BHR^t_{k,r,w,\nu, l, m,u}$, then the competitive ratio becomes smaller, in the case of polylogarithmic memory and $1<t\leq k/2$. At the same time, the competitive ratio is still larger than for the quantum online streaming algorithm. 
%\end{enumerate}

\subsubsection{Sublogarithmic Memory}
%{\bf Sublogarithmic Memory.}
We continue by analyzing the model with sublogarithmic memory.
Let us discuss the $PartialMOD_m^{\beta}$ function from \cite{AY12,agky14,agky16}. Feasible inputs for the problem are $X\in\{0,1\}^n$ such that $\#_1(X)=v\cdot 2^{\beta}$, where $\#_1(X)$ is the number of $1$s and $v\geq 2$. $PartialMOD_m^{\beta}(X)=v$ $\mod$ $2$.
It is known from \cite{AY12,agky14,agky16} that there is a quantum automaton that computes $PartialMOD_m^{\beta}$ using a single qubit and has not error. At the same time, any deterministic or probabilistic automaton and id-OBDDs requires $2^{\beta}$ states (width). Hence, we have the following result due to Lemma \ref{lm:rel-obdd-sa}.

\begin{lemma}\label{lm:bhp}
1. There is a quantum automaton that computes exactly $PartialMOD^{\beta}_m$ using single qubit;
%\item 
2. There is no probabilistic automaton that computes $PartialMOD^{\beta}_m$  with bounded error using less than $\beta$ bits.
%\end{itemize}
\end{lemma}

Let us apply the Black Hats Method to $f=PartialMOD_m^{\beta}$.

\begin{theorem}\label{th:pmod_results}
Suppose $P^t=BHM^t_{k,r,w}=BH^t_{k,r,w}(PartialMOD^{\beta}_m)$, $t\in\{1,\dots,k\}$, $\beta=O(\log n)$, $k=o(\log n), \beta \cdot k<\log_2 n$; then
%\begin{enumerate}

%\item  
1. There is no $c$-competitive deterministic online streaming algorithm with $s<\beta$ bits of memory that solves $P^t$, where $c<{\cal C}_1=\frac{w}{r}$.

%
%\item 
2. There is no deterministic online algorithm with unlimited computational power computing $P^1$ that is $c$-competitive, for $c<{\cal C}_1=w/r$. 

%\item 
3. There is no $c$-competitive in expectation randomized online streaming algorithm with $s<\beta$ bits of memory that solves $P^t$, where $c<{\cal C}_2=2^{-z}+(1-2^{-z})\frac{w}{r}$.

%\item 
4. There is a quantum online streaming algorithm that uses single qubit and solves $P^t$. The algorithm $Q$ is $c$-competitive in expectation, where

$c\leq \left(\left(1-\nu^2\right)^{z-1}\cdot 0.5 \cdot (r-w) + w\right)/r<{\cal C}_1,{\cal C}_2$.
%\end{enumerate}
\end{theorem}
\begin{proof}
Due to Lemma \ref{lm:bhp}, there is no probabilistic automaton that computes $PartialMOD^{\beta}_m$ with bounded error in a case of $s<\beta$ bits of memory. Therefore, there is no deterministic automaton that computes $PartialMOD^{\beta}_m$ in a case of $s<\beta$ bits of memory. So, if we use this property and Theorems \ref{th:bh-lower}, \ref{th:bh-rand-lower} and \ref{th:bh-general-d-lower}, then we obtain the claims 1, 2 and 3.  
  
Let us prove the Claim 4. Here we claim, that we can construct algorithm with single qubit, but not two qubits as in Theorem \ref{th:pqalgorithm}.  
  The idea of the algorithm is based on ides from \cite{agky14,agky16,AY12,kkm2018}.
Let us describe an algorithm $B$ for $BH^t_{k,r,w}(f)$, $f=PartialMOD_m^\beta$. Let an angle $\xi=\pi/2^{\beta+1}$.

{\bf Step $1$.} The algorithm emulates guessing for $g_1=\bigoplus_{j=1}^{k}f(X^j)$. $B$ starts on a state $|\psi\rangle=\frac{1}{\sqrt{2}}|0\rangle+\frac{1}{\sqrt{2}}|1\rangle$. The algorithm measures the qubit $|\psi\rangle$ before reading any input variables. $B$ gets $|0\rangle$ or $|1\rangle$ with equal probabilities. The result of the measurement is $y_1$.

{\bf Step $2$.} The algorithm reads $X^1$.  Algorithm $B$ rotates the qubit by the angle $\xi$ if the algorithm meets $1$. It  does nothing for $0$.

{\bf Step $3$.} If $B$ meets $2$, then it measures the qubit $|\psi\rangle=a|0\rangle+b|1\rangle$. If $PartialMOD_{m_1}^\beta(X^1)=1$ then the qubit is rotated by an angle $\pi/2+v\cdot \pi$, for some integer $v$, else the qubit is rotated by an angle $u\cdot \pi$, for some integer $u$. If $y_1=1$, then $a\in\{1,-1\}$ and $b=0$. If $y_1=0$, then $a=0$ and $b\in\{1,-1\}$. The result of the measurement for the qubit $|\psi\rangle$ is $y_2$. 

{\bf Step $4$.} The algorithm reads $X^2$ and does the similar action as in Step 2.

{\bf Step $5$.} If $B$ meets $2$ then it measures the qubit $|\psi\rangle=a|0\rangle+b|1\rangle$. If $f(X^2)=PartialMOD_{m_2}^\beta(X^2)=1$, then the qubit is rotated by an angle $\pi/2+v\cdot \pi$, for some integer $v$, else the qubit is rotated by an angle $u\cdot \pi$, for some integer $u$. Note that  before Step 4 if $y_2=1$, then $|\psi\rangle=|1\rangle$; and if $y_2=0$ then $|\psi\rangle=|0\rangle$. Therefore, if $y_3=PartialMOD_{m_2}^\beta(X^2)\oplus y_2 =1$, then $a\in\{1,-1\}$ and $b=0$. If $y_3=PartialMOD_{m_2}^\beta(X^2)\oplus y_2 =0$, then $a=0$ and $b\in\{1,-1\}$. The algorithm measures $|\psi\rangle$ and outputs $y_3 = PartialMOD_{m_2}^\beta(X^2)\oplus y_2$. 

{\bf Step $i$.} The step is similar to Step 4, but the algorithm reads $X^{i-1}$ and calculates $PartialMOD_{m_2}^\beta(X^{i-1})$.

{\bf Step $i+1$.} The step is similar to Step 5, but the algorithm outputs $y_{i} = PartialMOD_{m_2}^\beta(X^{i-1})\oplus y_{i-1}$.

{\bf Step $2k+2$.} The algorithm reads and skips the last part of the input. $B$ does not need these variables, because it guesses $y_1$ and using this value we already can obtain $y_{2},...,y_{k}$ without $X^{k}$.
\Endproof
\end{proof}

This theorem gives us the following important results.
%\begin{enumerate}
%\item 
%
There is a quantum online streaming algorithm with one qubit of memory for  $BHM^t_{k,r,w}$ having a better competitive ratio than
 any classical (deterministic or randomized) online streaming algorithm with sublogarithmic size of  memory.

%\end{enumerate}
%\noindent
\paragraph*{Acknowledgements.}
The research is supported by Russian Science Foundation Grant 17-71-10152.

%This work was supported by ERC Advanced Grant MQC. The work is performed according to the Russian Government Program of Competitive Growth of Kazan Federal University. A part of the research work was done while M. Ziatdinov was visiting University of Latvia in October 2017.
%
We thank Andris Ambainis, Alexanders Belovs and Abuzer Yakarilmaz from University of Latvia for helpful discussions.
\bibliographystyle{plain}
\bibliography{tcs}

\begin{thebibliography}{10}

\bibitem{aakv2018}
F.~Ablayev, M.~Ablayev, K.~Khadiev, and A.~Vasiliev.
\newblock Classical and quantum computations with restricted memory.
\newblock {\em LNCS}, 11011:129--155, 2018.

\bibitem{aakk2018}
F.~Ablayev, A.~Ambainis, K.~Khadiev, and A.~Khadieva.
\newblock Lower bounds and hierarchies for quantum memoryless communication
  protocols and quantum ordered binary decision diagrams with repeated test.
\newblock {\em In SOFSEM, LNCS}, 10706:197--211, 2018.

\bibitem{agkmp2005}
F.~Ablayev, A.~Gainutdinova, M.~Karpinski, C.~Moore, and C.~Pollett.
\newblock On the computational power of probabilistic and quantum branching
  program.
\newblock {\em Information and Computation}, 203(2):145--162, 2005.

\bibitem{agky16}
F.~Ablayev, A.~Gainutdinova, K.~Khadiev, and A.~Yakary{\i}lmaz.
\newblock Very narrow quantum \mbox{OBDD}s and width hierarchies for classical
  \mbox{OBDD}s.
\newblock {\em Lobachevskii Journal of Mathematics}, 37(6):670--682, 2016.

\bibitem{agky14}
F.~Ablayev, A.~Gainutdinova, K.~Khadiev, and A.~Yakaryılmaz.
\newblock Very narrow quantum \mbox{OBDD}s and width hierarchies for classical
  \mbox{OBDD}s.
\newblock In {\em DCFS}, volume 8614 of {\em LNCS}, pages 53--64. Springer,
  2014.

\bibitem{akv2008}
F.~Ablayev, A.~Khasianov, and A.~Vasiliev.
\newblock On complexity of quantum branching programs computing equality-like
  boolean functions.
\newblock {\em ECCC}, 2010.

\bibitem{AY12}
A.~Ambainis and A.~Yakary{\i}lmaz.
\newblock Superiority of exact quantum automata for promise problems.
\newblock {\em Information Processing Letters}, 112(7):289--291, 2012.

\bibitem{AY15}
Andris Ambainis and Abuzer Yakary{\i}lmaz.
\newblock Automata and quantum computing.
\newblock Technical Report 1507.01988, arXiv, 2015.

\bibitem{bk2009}
L.~Becchetti and E.~Koutsoupias.
\newblock Competitive analysis of aggregate max in windowed streaming.
\newblock In {\em ICALP}, volume 5555 of {\em LNCS}, pages 156--170, 2009.

\bibitem{blm2015}
Joan Boyar, Kim~S Larsen, and Abyayananda Maiti.
\newblock The frequent items problem in online streaming under various
  performance measures.
\newblock {\em International Journal of Foundations of Computer Science},
  26(4):413--439, 2015.

\bibitem{Fre79}
R\={u}si\c{n}\v{s} Freivalds.
\newblock Fast probabilistic algorithms.
\newblock In {\em Mathematical Foundations of Computer Science 1979}, volume~74
  of {\em LNCS}, pages 57--69, 1979.

\bibitem{g15}
A.~Gainutdinova.
\newblock Comparative complexity of quantum and classical \mbox{OBDD}s for
  total and partial functions.
\newblock {\em Russian Mathematics}, 59(11):26--35, 2015.

\bibitem{gy2015}
A.~Gainutdinova and A.~Yakary{\i}lmaz.
\newblock Unary probabilistic and quantum automata on promise problems.
\newblock In {\em Developments in Language Theory}, pages 252--263. Springer,
  2015.

\bibitem{gy2017}
A.~Gainutdinova and A.~Yakary{\i}lmaz.
\newblock Nondeterministic unitary \mbox{OBDD}s.
\newblock In {\em CSR 2017}, pages 126--140. Springer, 2017.

\bibitem{gy2018}
A.~Gainutdinova and A.~Yakary{\i}lmaz.
\newblock Unary probabilistic and quantum automata on promise problems.
\newblock {\em Quantum Information Processing}, 17(2):28, 2018.

\bibitem{gkkrw2007}
D.~Gavinsky, J.~Kempe, I.~Kerenidis, R.~Raz, and R.~de~Wolf.
\newblock Exponential separations for one-way quantum communication complexity,
  with applications to cryptography.
\newblock In {\em STOC '07}, pages 516--525, 2007.

\bibitem{gk2015}
Y.~Giannakopoulos and E.~Koutsoupias.
\newblock Competitive analysis of maintaining frequent items of a stream.
\newblock {\em Theoretical Computer Science}, 562:23--32, 2015.

\bibitem{kmrs86}
A.~R Karlin, M.~S Manasse, L.~Rudolph, and D.~D Sleator.
\newblock Competitive snoopy caching.
\newblock In {\em FOCS, 1986., 27th Annual Symposium on}, pages 244--254. IEEE,
  1986.

\bibitem{kk2017}
K.~Khadiev and A.~Khadieva.
\newblock Reordering method and hierarchies for quantum and classical ordered
  binary decision diagrams.
\newblock In {\em CSR 2017}, volume 10304 of {\em LNCS}, pages 162--175.
  Springer, 2017.

\bibitem{kkm2018}
K.~Khadiev, A.~Khadieva, and I.~Mannapov.
\newblock Quantum online algorithms with respect to space and advice
  complexity.
\newblock {\em Lobachevskii Journal of Mathematics}, 39(9):1210--1220, 2018.

\bibitem{l2006}
Fran\c{c}ois Le~Gall.
\newblock Exponential separation of quantum and classical online space
  complexity.
\newblock SPAA '06, pages 67--73. ACM, 2006.

\bibitem{l2009}
Fran{\c{c}}ois Le~Gall.
\newblock Exponential separation of quantum and classical online space
  complexity.
\newblock {\em Theory of Computing Systems}, 45(2):188--202, 2009.

\bibitem{nc2010}
Michael~A Nielsen and Isaac~L Chuang.
\newblock {\em Quantum computation and quantum information}.
\newblock Cambridge university press, 2010.

\bibitem{ss2005}
M.~Sauerhoff and D.~Sieling.
\newblock Quantum branching programs and space-bounded nonuniform quantum
  complexity.
\newblock {\em Theoretical Computer Science}, 334(1):177--225, 2005.

\bibitem{st85}
Daniel~D Sleator and Robert~E Tarjan.
\newblock Amortized efficiency of list update and paging rules.
\newblock {\em Communications of the ACM}, 28(2):202--208, 1985.

\bibitem{Weg00}
Ingo Wegener.
\newblock {\em Branching Programs and Binary Decision Diagrams: Theory and
  Applications}.
\newblock SIAM, 2000.

\bibitem{y2009}
Q.~Yuan.
\newblock {\em Quantum online algorithms}.
\newblock UC Santa Barbara, 2009.
\newblock PhD thesis.

\end{thebibliography}
%\newpage
\appendix
\section{Definitions of OBDDs}\label{apx:obdd}

OBDD is a restricted version of a branching program (BP). BP over a set $X$ of $n$ Boolean variables is a directed acyclic graph with two distinguished nodes $s$ (a source node) and $t$ (a sink node). We denote it $P_{s,t}$ or just $P$. Each inner node $v$ of $P$ is associated with a variable $x\in X$. A {\em deterministic} BP has exactly two outgoing edges labeled $x=0$ and $x=1$ respectively for each node $v$. The program $P$ computes a Boolean function $f(X)$ ($f:\{0,1\}^n \rightarrow \{0,1\}$) as follows: for each $\sigma\in\{0,1\}^n$ we let $f(\sigma)=1$ iff there exists at least one $s-t$ path (called {\em accepting} path for $\sigma$) such that all edges along this path are consistent with $\sigma$. A {\em size} of branching program $P$ is a number of nodes.
Ordered  Binary Decision Diagram (OBDD) is a BP with following restrictions: 

(i) Nodes can be partitioned into levels $V_1, \ldots, V_{\ell+1}$ such that  $s$ belongs to the first level  $V_1$ and sink node $t$ belongs to the last level $V_{\ell+1}$. Nodes from level $V_j$ have outgoing edges only to nodes of level $V_{j+1}$, for $j \le \ell$.

(ii)All inner nodes of one level are labeled by the same variable.

(iii)Each variable is tested on each path only once.

A {\em width} of a program $P$ is $ width(P)=\max_{1\le j\le \ell}|V_j|. $ 
OBDD $P$ reads variables in its individual  order
$\theta(P)=(j_1,\dots,j_n)$. Let $id=(1,\dots,n)$ be a natural order of input variables. If OBDD reads input variables in the order $id$, then we denote the model as id-OBDD.

Probabilistic OBDD (POBDD) can have more than two edges for a node, and we choose one of them using a probabilistic mechanism. POBDD $P$ computes a Boolean function $f$ with bounded error $0.5-\varepsilon$ if probability of the right answer is at least $0.5+\varepsilon$.

Let us define a quantum OBDD. It is given in different terms, but they are equivalent, see \cite{agkmp2005} for more details.
For a given $ n>0 $, a quantum OBDD $ P$ of width $d$ defined on $ \{0,1\}^n $, is a 4-tuple
$
    P=(T,|\psi\rangle_0,Accept,\pi),
$
where
%\begin{itemize}
%    \item 
$ T = \{ T_j : 1 \leq j \leq n \mbox{ and } T_j = (G_j^0,G_j^1) \} $ are ordered
pairs of (left) unitary matrices representing transitions. Here $ G_j^0 $ or
$ G_j^1 $ is applied on the $j$-th step. A choice is determined by the input bit.  
%    \item 
The vector $\ket{\psi}_0$ is a initial vector from the $ d $-dimensional Hilbert space over the field of complex numbers.  $ \ket{\psi}_0=\ket{q_0}$ where $ q_0 $ corresponds to the initial node.
%    \item 
$ Accept \subset \{1,\ldots,d\} $ is a set of accepting nodes.
%\item 
$ \pi $ is a permutation of $ \{1,\ldots,n\} $. It defines the order of input bits.
%\end{itemize}

  For any given input $ \nu \in \{0,1\}^n $, the computation of $P$ on $\nu$ can be traced by the $d$-dimensional vector from a Hilbert space over the field of complex numbers. The initial one is $ \ket{\psi}_0$. In each step $j$, $1 \leq j \leq n$, the input bit $ x_{\pi(j)} $ is tested and then the corresponding unitary operator is applied:  
$
    \ket{\psi}_j = G_j^{x_{\pi(j)}} (\ket{\psi}_{j-1}),
$ 
where  $ \ket{\psi}_j $ represents the state of the system after the $ j$-th step.
The quantum OBDD can measure one or more qubits on any steps. Let the program be in state $\ket{\psi}=(v_1,
\dots, v_d)$ before a measurement and let us measure the $i$-th qubit.
Let states with numbers $j^0_1,\dots,j^0_{d/2}$ correspond to the $0$ value of the $i$-th qubit, and states with numbers $j^1_1,\dots,j^1_{d/2}$ correspond to the $1$ value of the $i$-th qubit.
The result of the measurement of the $i$-th qubit is $1$ with probability $pr_1= \sum_{z=1}^{d/2}|v_{j^1_z}|^2$ and $0$ with probability $pr_0=1-pr_1$.
The program $P$ measures all qubits at the end of the computation process. The program  accepts the input $ \sigma $ (returns $1$ on the input) with probability $
    Pr_{accept}(\nu)=\sum_{i \in Accept} v^2_i
$, for $ \ket{\psi}_n=(v_1,\dots,v_d)$.

$P_{\varepsilon}(\nu)=1$ if $P$ accepts input $\nu\in\{0,1\}^n$ with  probability at least $ 0.5+\varepsilon$, and $P_{\varepsilon}(\nu)=0$ if  $P$ accepts the input $\nu\in\{0,1\}^n$ with  probability at most $ 0.5-\varepsilon$, for $ \varepsilon \in (0,0.5] $.
We say that a function $f$ is computed by $ P$ with a bounded error if there exists $ \varepsilon \in (0,0.5] $ such that $ P_{\varepsilon}(\nu)=f(\nu)$ for any  $\nu\in\{0,1\}^n$. We can say that $P$ computes $f$ with a bounded error $0.5-\varepsilon$. 

{\bf Automata.}
We can say that an automaton is an id-OBDD such that a transition function for each level is the same. Note that id-OBDD is OBDD with an order $id=(1,\dots,n)$.

\end{document}